\documentclass{article}

\usepackage[main, preprint]{neurips_2026}

\usepackage[utf8]{inputenc}
\usepackage[T1]{fontenc}
\usepackage{hyperref}
\usepackage{url}
\usepackage{booktabs}
\usepackage{amsfonts}
\usepackage{amsmath}
\usepackage{amssymb}
\usepackage{amsthm}
\newtheorem{proposition}{Proposition}
\usepackage{nicefrac}
\usepackage{microtype}
\usepackage{xcolor}
\usepackage{graphicx}
\usepackage{multirow}
\usepackage{subcaption}
\usepackage{acro}

\DeclareAcronym{cp}{short=CP, long=conformal prediction}
\DeclareAcronym{mlp}{short=MLP, long=multilayer perceptron}
\DeclareAcronym{cdf}{short=CDF, long=cumulative distribution function}
\DeclareAcronym{atp}{short=ATP, long=American Trends Panel}

\graphicspath{{../figures/publication-ready/}}

\title{Socio-Conformal Calibration in Complex Survey Data:\\Marginal Validity Is Not Enough for Subgroup Reliability}

\author{%
Amir Rafe, Ph.D.\\
Texas State University\\
San Marcos, USA\\
\texttt{amir.rafe@txstate.edu}\\
ORCID: \href{https://orcid.org/0000-0002-4089-2088}{0000-0002-4089-2088}
\And
Subasish Das, Ph.D.\\
Texas State University\\
San Marcos, USA\\
\texttt{subasish@txstate.edu}\\
ORCID: \href{https://orcid.org/0000-0002-1671-2753}{0000-0002-1671-2753}
}

\begin{document}

\maketitle

\begin{abstract}
Machine-learning systems used in survey-based social measurement require uncertainty estimates that are reliable across population subgroups, not merely valid in aggregate. We study ordinal conformal prediction for five-level AI-attitude forecasting on the Pew American Trends Panel (Wave~152; $n{=}4{,}591$; 12 race$\times$education subgroups), comparing standard split conformal, Mondrian (group-specific) conformal, and a regularized Mondrian comparator across 100 respondent-disjoint splits with survey-weighted evaluation. Standard conformal achieves nominal marginal coverage for all four base predictors but leaves weighted subgroup gaps of ${\sim}13$~percentage points. For the strongest predictor (XGBoost), Mondrian worsens the fairness-efficiency trade-off: weighted set size rises by $+0.036$ ($d_z{=}1.66$) while the weighted subgroup gap grows by $+0.013$ ($d_z{=}0.30$). A regularized comparator that shrinks group thresholds toward the global quantile mitigates this instability ($\Delta$~gap~$= -0.001$, $\Delta$~size~$= +0.012$) but does not yield a decisive fairness gain. Failure analysis traces the mechanism to calibration-cell fragmentation interacting with group-specific confidence mismatch. The negative result persists across alternate outcome codings and subgroup granularities, demonstrating that nominal marginal validity is insufficient for subgroup reliability and that naive group-specific calibration is not a dependable fairness remedy in complex survey settings.
\end{abstract}

\acresetall

\section{Introduction}
\label{sec:intro}

\Ac{cp} has emerged as a leading framework for distribution-free uncertainty quantification in machine learning~\citep{vovk2005algorithmic, shafer2008tutorial, papadopoulos2002inductive}. By providing finite-sample coverage guarantees with minimal distributional assumptions, conformal methods ground predictions in rigorous statistical validity, a property increasingly demanded by AI governance frameworks. Recent tutorials~\citep{angelopoulos2023conformal} have accelerated adoption, and extensions now address quantile regression~\citep{romano2019conformalized}, distribution shift~\citep{tibshirani2019conformal, barber2023conformal}, and adaptive online settings~\citep{gibbs2021adaptive}. For survey-based social measurement, where predictions inform policy discussions about public attitudes, reliable uncertainty is not optional: prediction sets that over-cover some demographic groups while under-covering others undermine the scientific basis for equitable policy.

Standard split \ac{cp} guarantees \emph{marginal} coverage: the prediction set contains the true response with probability at least $1 - \alpha$ on average over the test population. However, this aggregate guarantee can mask substantial disparities across demographic subgroups. In complex surveys with design-based weights and intersectional demographic cells, thin calibration cells (some with fewer than 50 observations) create settings where marginal validity and subgroup reliability diverge sharply.

The canonical theoretical remedy is Mondrian \ac{cp}~\citep{vovk2005algorithmic, vovk2012conditional}, which estimates group-specific thresholds to achieve conditional coverage within each predefined subgroup. \citet{romano2020malice} formalize equalized coverage objectives, and \citet{ding2023class} extend Mondrian calibration to the many-class regime, showing that clustered thresholds can recover conditional validity when individual class cells are thin. More recent work proposes adaptive group selection for equalized coverage~\citep{zhou2024conformal}, bridges fairness and efficiency through surrogate-assisted clustering~\citep{gao2025bridging}, and studies multi-group learning for conformal validity~\citep{deng2023group}. These advances notwithstanding, conditional validity beyond predefined groups faces fundamental impossibility results~\citep{lei2014distribution, foygel2021limits}, motivating alternative approaches through multivalid~\citep{jung2023batch, hebert2018multicalibration} and conditional~\citep{gibbs2025conformal} guarantees. Critically, \citet{cresswell2024conformal} demonstrate that \ac{cp} sets can themselves cause disparate impact, underscoring that nominal fairness corrections may introduce new equity concerns.

Despite this theoretical progress, the empirical behavior of Mondrian \ac{cp} in realistic survey settings remains poorly characterized. In such settings, calibration cells may contain as few as 32 observations, outcome distributions vary across groups, and design-based weights create effective sample sizes smaller than nominal counts. Whether Mondrian's theoretical benefits translate into practical fairness improvements under these conditions is an open question.

Our empirical setting centers on public attitudes toward artificial intelligence, a domain where demographic variation in trust and perceived impact is well documented~\citep{zhang2019artificial, brauner2023public, novozhilova2024looking}. Trust in AI systems varies across age, education, and racial groups~\citep{shin2021effects, glikson2020human, neudert2020global}, making intersectional subgroup analysis essential for responsible deployment of predictive models in this space. Ordinal response scales are standard in survey measurement~\citep{agresti2010analysis, mccullagh1980regression}, and \ac{cdf}-based nonconformity scores extend \ac{cp} to ordered outcomes~\citep{lu2022fair, gupta2022nested}. Risk-controlling prediction sets~\citep{bates2021distribution} provide a general framework for structured losses. Complex survey designs require design-based weights for population-representative inference~\citep{lumley2010complex}, and \ac{cp} under covariate shift~\citep{tibshirani2019conformal} provides the closest theoretical connection, though formal guarantees for survey-weighted conformal calibration remain an open problem. The trade-off between fairness criteria is well established~\citep{chouldechova2017fair, kleinberg2017inherent}, and intersectional approaches~\citep{kearns2018preventing} together with calibration-fairness connections~\citep{pleiss2017calibration} inform our subgroup analysis.

We address this gap with a 100-split empirical study of ordinal \ac{cp} on the Pew \ac{atp} (Wave~152, $n{=}4{,}591$), comparing standard, Mondrian, and regularized Mondrian calibration across 12 race$\times$education subgroups with survey-weighted evaluation. Our study is organized around five research questions: \textbf{RQ1.}~Does standard \ac{cp} achieve nominal marginal coverage? \textbf{RQ2.}~Do subgroup disparities persist under standard \ac{cp}? \textbf{RQ3.}~Do group-aware threshold variants reduce subgroup coverage gaps, and at what efficiency cost? \textbf{RQ4.}~How much do survey weights shift the fairness picture? \textbf{RQ5.}~How much does predictor choice affect set efficiency once coverage is calibrated?

This paper makes three contributions. First, we provide a 100-split verified empirical study of ordinal \ac{cp} in a complex survey setting with weighted subgroup evaluation, showing that nominal marginal validity coexists with meaningful subgroup disparities across race$\times$education cells. Second, our central negative result demonstrates that vanilla Mondrian worsens the fairness-efficiency trade-off for the strongest practical predictor. Third, we introduce a regularized Mondrian comparator based on James-Stein shrinkage~\citep{james1961estimation, efron1975data} and present an explicit failure analysis identifying calibration-cell fragmentation and group-specific confidence mismatch as the mechanisms behind Mondrian's instability.

\section{Data}
\label{sec:data}

We use the Pew Research Center \ac{atp}, Wave~152 (August 2024)~\citep{pewresearch2024w152}. The \ac{atp} is a nationally representative online panel of U.S.\ adults recruited via address-based sampling with design-based weights calibrated to the U.S.\ adult population. Our analytic sample comprises $n{=}4{,}591$ respondents with valid outcome and subgroup labels.

\textbf{Outcome.} The primary outcome is a five-level ordinal scale of expected AI impact (\textsc{aichange\_w152}): \emph{Very positive} (1), \emph{Somewhat positive} (2), \emph{Neither positive nor negative} (3), \emph{Somewhat negative} (4), and \emph{Very negative} (5). ``Not sure'' ($n{=}757$) and refusal responses are excluded from the primary analysis and tested in sensitivity branches (Section~\ref{sec:design}).

\textbf{Predictors.} We use 66 survey covariates drawn from five thematic families: demographics (29 items), concern and excitement toward AI (27 items), AI familiarity and use (6 items), regulation and confidence (3 items), and personal benefit/harm (1 item). All predictors are contemporaneous Wave~152 variables; items that directly restate the outcome are excluded.

\textbf{Protected groups.} We define $G{=}12$ predefined subgroups as the full intersection of four race/ethnicity categories and three education levels. Table~\ref{tab:groups} shows the resulting cross-tabulation. Cell sizes span an order of magnitude: the thinnest cell (Asian/Other $|$ College+) yields only 32 calibration observations per split, while the largest (White non-Hispanic $|$ HS or less) yields 355. This imbalance creates the thin-cell conditions central to our analysis.

\begin{table}[h]
  \caption{\textbf{Protected group cross-tabulation.} $n$ = unweighted count; parenthetical = mean calibration-cell size $n_\mathrm{cal}$ per 30\% split. \colorbox{red!12}{Shaded} cells have $n_\mathrm{cal} < 50$, the regime where Mondrian threshold estimation becomes unstable.}
  \label{tab:groups}
  \centering
  \small
  \begin{tabular}{l rrr r}
    \toprule
    & \multicolumn{3}{c}{Education} \\
    \cmidrule(lr){2-4}
    Race / Ethnicity & College+ & Some college & HS or less & Total \\
    \midrule
    White non-Hispanic  & 664\,(199)  & 818\,(245) & 1{,}184\,(355) & 2{,}666 \\
    Hispanic            & 209\,(63)   & 266\,(80)  & 232\,(70)      & 707 \\
    Black non-Hispanic  & 151\,\colorbox{red!12}{(45)}  & 175\,(53)  & 198\,(59)      & 524 \\
    Asian / Other       & 105\,\colorbox{red!12}{(32)}  & 144\,\colorbox{red!12}{(43)}  & 445\,(134)     & 694 \\
    \midrule
    Total               & 1{,}129     & 1{,}403    & 2{,}059        & 4{,}591 \\
    \bottomrule
  \end{tabular}
\end{table}

\textbf{Survey weights.} The Pew-provided design-based weight (\textsc{weight\_w152}) is used for all population-representative evaluation summaries. Weights are not used inside the conformal calibration step except in the weighted-threshold sensitivity variant.

\textbf{Notation.} Let $\mathcal{D} = \{(X_i, Y_i, W_i, G_i)\}_{i=1}^n$ where $X_i \in \mathbb{R}^{66}$ is the covariate vector, $Y_i \in \{1,\dots,5\}$ the ordinal response, $W_i > 0$ the survey weight, and $G_i \in \{1,\dots,12\}$ the group label. Each random split partitions $\mathcal{D}$ into respondent-disjoint training ($|I_\mathrm{tr}|{=}1{,}834$), calibration ($|I_\mathrm{cal}|{=}1{,}383$), and test ($|I_\mathrm{te}|{=}1{,}374$) sets with guarded stratification across outcome$\times$group cells.

\section{Methods}
\label{sec:methods}

Figure~\ref{fig:overview} summarizes the full experimental pipeline. Survey data are split into train, calibration, and test sets across 100 random partitions; four base predictors produce ordinal \ac{cdf} scores; standard, Mondrian, and regularized Mondrian conformal methods calibrate prediction sets at $\alpha{=}0.10$; and all evaluation metrics are survey-weighted.

\begin{figure}[t]
  \centering
  \includegraphics[width=\textwidth]{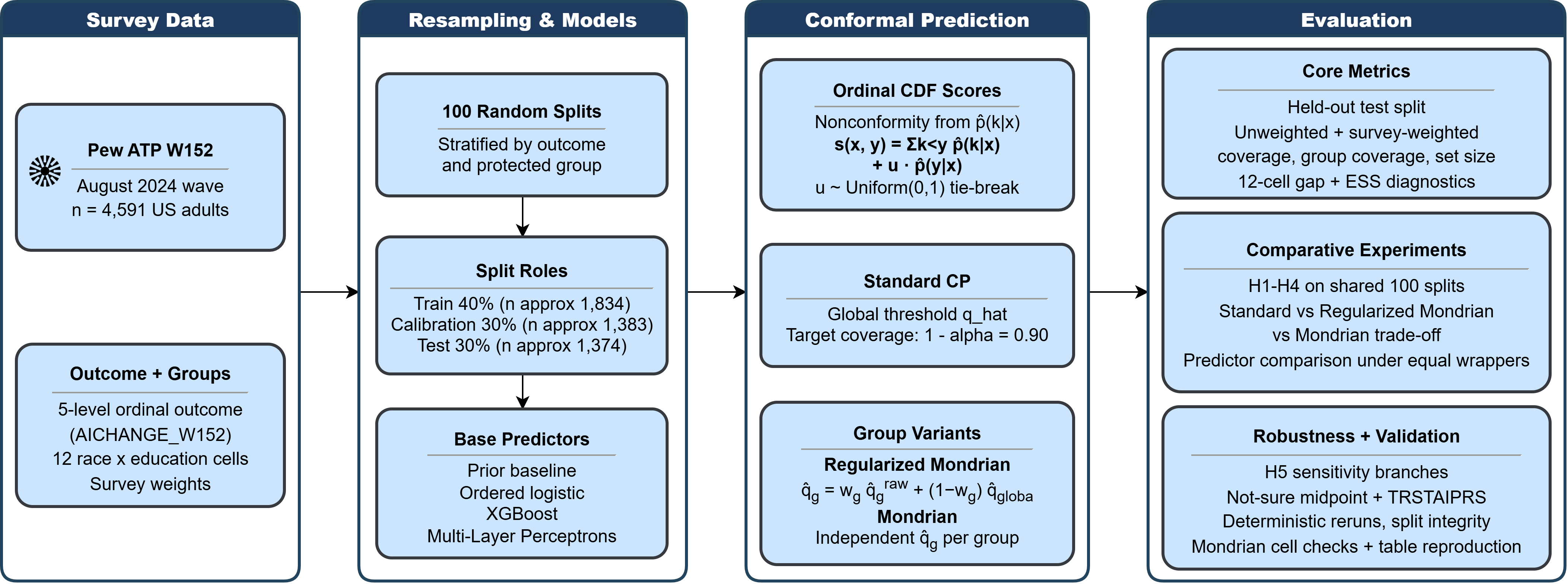}
  \caption{\textbf{Experimental pipeline.} Survey data ($n{=}4{,}591$) are split into train/calibration/test sets across 100 random partitions. Four base predictors produce ordinal \acs{cdf} scores. Standard, Mondrian, and regularized Mondrian conformal methods calibrate prediction sets at $\alpha{=}0.10$. All evaluation metrics are survey-weighted.}
  \label{fig:overview}
\end{figure}

\subsection{Base predictors}
\label{sec:base}

Each base predictor produces estimated class probabilities $\hat{p}_c(x) \approx \Pr(Y{=}c \mid X{=}x)$ for $c = 1,\dots,5$, yielding a fitted \ac{cdf} $\hat{F}(y \mid x) = \sum_{c \le y} \hat{p}_c(x)$. We compare four models spanning calibration quality and capacity. The prior baseline uses marginal class frequencies from the training set with no covariates. The ordered logistic model~\citep{mccullagh1980regression} is a proportional-odds specification with 14 demographic and attitudinal features. XGBoost ~\citep{chen2016xgboost} uses gradient-boosted trees with all 66 features and serves as the strongest practical predictor. A two-layer \ac{mlp} also uses all 66 features. All models are trained on $I_\mathrm{tr}$ only; their predicted probabilities define the nonconformity scores used for calibration on $I_\mathrm{cal}$.

\subsection{Conformal calibration}
\label{sec:conformal}

\textbf{Nonconformity score.}
\label{sec:score}
For ordinal outcomes we use a randomized \ac{cdf}-based score that preserves the label ordering:
\begin{equation}
  s(x, y; U) = 1 - \hat{F}(y \mid x) + U \cdot \hat{p}_y(x), \qquad U \sim \mathrm{Unif}(0,1).
  \label{eq:score}
\end{equation}
Smaller scores indicate labels more compatible with the fitted distribution. The uniform variate $U$ breaks ties induced by discrete probabilities. For each calibration observation $i \in I_\mathrm{cal}$, we compute $S_i = s(X_i, Y_i; U_i)$ with an independent draw $U_i$.

\textbf{Standard split conformal.}
\label{sec:standard}
The global threshold is
\begin{equation}
  \hat{q}_\alpha = \mathrm{Quantile}_{1-\alpha}\!\bigl(\{S_i : i \in I_\mathrm{cal}\} \cup \{\infty\}\bigr),
  \label{eq:standard_threshold}
\end{equation}
and the prediction set for a new test point $x$ is $\hat{C}_\alpha(x) = \{y \in \{1,\dots,5\} : s(x, y; U) \le \hat{q}_\alpha\}$.
Under exchangeability of calibration and test data, this guarantees marginal coverage $\Pr(Y_{n+1} \in \hat{C}_\alpha(X_{n+1}) \mid \mathcal{D}_\mathrm{tr}) \ge 1 - \alpha$~\citep{vovk2005algorithmic, shafer2008tutorial}. We set $\alpha = 0.10$ throughout, targeting 90\% nominal coverage.

\textbf{Mondrian conformal.}
\label{sec:mondrian}
Mondrian \ac{cp}~\citep{vovk2005algorithmic, vovk2012conditional} replaces the global threshold with group-specific thresholds. Let $I_\mathrm{cal}^{(g)} = \{i \in I_\mathrm{cal} : G_i = g\}$ with $n_g = |I_\mathrm{cal}^{(g)}|$. The Mondrian threshold for group $g$ is
\begin{equation}
  \hat{q}_{g,\alpha} = \mathrm{Quantile}_{1-\alpha}\!\bigl(\{S_i : i \in I_\mathrm{cal}^{(g)}\} \cup \{\infty\}\bigr),
  \label{eq:mondrian_threshold}
\end{equation}
and the prediction set becomes $\hat{C}^\mathrm{Mond}_\alpha(x) = \{y : s(x,y;U) \le \hat{q}_{g(x),\alpha}\}$. Under within-group exchangeability this provides group-conditional coverage $\Pr(Y_{n+1} \in \hat{C}^\mathrm{Mond}_\alpha(X_{n+1}) \mid G_{n+1}{=}g, \mathcal{D}_\mathrm{tr}) \ge 1 - \alpha$. In our setting, $n_g$ ranges from 32 to 355 (Table~\ref{tab:groups}), creating substantial variance in threshold estimates for thin cells.

\textbf{Regularized Mondrian.}
\label{sec:regmond}
To mitigate thin-cell instability, we introduce a regularized variant that shrinks each group threshold toward the global value:
\begin{equation}
  \hat{q}_{g,\alpha}^\mathrm{reg} = w_g \cdot \hat{q}_{g,\alpha} + (1 - w_g) \cdot \hat{q}_\alpha, \qquad w_g = \frac{n_g}{n_g + \lambda},\quad \lambda = 50.
  \label{eq:reg_threshold}
\end{equation}
For large cells ($n_g \gg \lambda$), $w_g \to 1$ and the threshold is nearly group-specific; for small cells ($n_g \ll \lambda$), $w_g \to 0$ and it falls back to the global value. Shrinkage weights in our data range from 0.39 to 0.88. This construction is motivated by James-Stein shrinkage~\citep{james1961estimation, efron1975data} and partial pooling in hierarchical models. We position it as an empirical comparator, not a theorem-backed contribution.

\subsection{Survey-weighted evaluation metrics}
\label{sec:metrics}

All evaluation uses the design-based survey weight $W_i$. We define four metrics. Weighted marginal coverage:
\begin{equation}
  \widehat{\mathrm{Cov}}_w = \frac{\sum_{i \in I_\mathrm{te}} W_i \cdot \mathbf{1}\{Y_i \in \hat{C}_i\}}{\sum_{i \in I_\mathrm{te}} W_i}.
  \label{eq:wtd_cov}
\end{equation}
Weighted subgroup coverage for group $g$:
\begin{equation}
  \widehat{\mathrm{Cov}}_{w,g} = \frac{\sum_{i \in I_\mathrm{te}} W_i \cdot \mathbf{1}\{G_i=g\} \cdot \mathbf{1}\{Y_i \in \hat{C}_i\}}{\sum_{i \in I_\mathrm{te}} W_i \cdot \mathbf{1}\{G_i=g\}}.
  \label{eq:wtd_group_cov}
\end{equation}
Weighted subgroup gap, our primary fairness summary:
\begin{equation}
  \widehat{\mathrm{Gap}}_w = \max_g \widehat{\mathrm{Cov}}_{w,g} - \min_g \widehat{\mathrm{Cov}}_{w,g}.
  \label{eq:wtd_gap}
\end{equation}
Weighted average set size, our efficiency measure:
\begin{equation}
  \widehat{\mathrm{Size}}_w = \frac{\sum_{i \in I_\mathrm{te}} W_i \cdot |\hat{C}_i|}{\sum_{i \in I_\mathrm{te}} W_i}.
  \label{eq:wtd_size}
\end{equation}

\textbf{Statistical inference.} We run 100 random respondent-disjoint splits. For each split $r$, we compute paired deltas $\Delta_r = \mathrm{metric}^{(\mathrm{method})}(r) - \mathrm{metric}^{(\mathrm{standard})}(r)$. We report means with 95\% confidence intervals $\bar{\Delta} \pm 1.96 \cdot \mathrm{SE}(\Delta)$ and Cohen's $d_z = \bar{\Delta} / \mathrm{SD}(\Delta)$ as a paired effect size.

\subsection{Experimental protocol}
\label{sec:design}

Our five confirmatory questions (RQ1--RQ5) are stated in Section~\ref{sec:intro}. We generate 100 random respondent-disjoint 40/30/30 partitions with guarded stratification across outcome$\times$group cells; no respondent appears in more than one partition within any split. Two sensitivity branches test robustness to outcome coding: (A)~merge ``Not sure'' responses into the ordinal midpoint ($n{=}5{,}348$), and (B)~use the binary direct-trust item \textsc{trstaiprs\_w152} ($n{=}4{,}071$). An alternate-group extension evaluates race-only (4 groups) and education-only (3 groups) subgroup families.

\section{Results}
\label{sec:results}

\subsection{Main results: marginal validity and the Mondrian trade-off}
\label{sec:res_marginal}

Table~\ref{tab:main} summarizes the main results. Under standard split conformal calibration, all four base predictors achieve weighted marginal coverage within 1~percentage point of the 90\% nominal target (range: 0.900 for ordered logistic to 0.910 for the prior baseline), confirming marginal validity across 100 independent splits. However, weighted subgroup gaps remain substantial: the max$-$min gap across the 12 race$\times$education cells ranges from 0.130 (\acs{mlp}) to 0.152 (ordered logistic). For XGBoost, the weighted subgroup gap is 0.135, meaning the best- and worst-covered subgroups differ by 13.5~percentage points despite aggregate validity.

For XGBoost, Mondrian conformal does not reduce the weighted subgroup gap. Instead, it increases both the gap (from 0.135 to 0.147, $\Delta = {+}0.013$) and the weighted set size (from 3.240 to 3.275, $\Delta = {+}0.036$). Figure~\ref{fig:tradeoff} visualizes this trade-off across all model$\times$method combinations. The pattern is consistent for the three informative predictors (XGBoost, \acs{mlp}, ordered logistic): Mondrian produces the largest weighted gap \emph{and} the largest set size within each model family. The prior baseline is an exception (see Discussion), where both Mondrian variants reduce gap and set size. The set-size increase is statistically significant with a large paired effect size ($d_z = 1.66$, $p < 0.001$); the gap increase is directionally consistent but noisier ($d_z = 0.30$). Full paired deltas with confidence intervals appear in Table~\ref{tab:paired} (Appendix~\ref{app:paired}).

The regularized Mondrian comparator ($\lambda = 50$) mitigates both dimensions: the weighted gap decreases slightly relative to standard ($\Delta = -0.001$) with a smaller set-size penalty ($\Delta = {+}0.012$). For the prior baseline, where the underlying predictor is uninformative, regularized Mondrian achieves the smallest weighted gap of any configuration (0.105).

\begin{table}[t]
  \caption{\textbf{Main results.} 100-split survey-weighted means ($\alpha = 0.10$). Wtd.\ Gap $= \max_g \widehat{\mathrm{Cov}}_{w,g} - \min_g \widehat{\mathrm{Cov}}_{w,g}$. Best gap within XGBoost in \textbf{bold}; XGBoost standard shown with $\diamond$ as the reference configuration used in the main text.}
  \label{tab:main}
  \centering
  \small
  \begin{tabular}{llccc}
    \toprule
    Model & Method & Wtd.\ Coverage & Wtd.\ Set Size & Wtd.\ Gap \\
    \midrule
    \multirow{3}{*}{XGBoost}
      & Standard$^{\diamond}$ & 0.905 & 3.240 & 0.135 \\
      & Mondrian & 0.908 & 3.275 & 0.147 \\
      & Reg.\ Mondrian & 0.906 & 3.251 & \textbf{0.133} \\
    \midrule
    \acs{mlp} & Standard & 0.904 & 3.279 & 0.130 \\
    Ordered logistic & Standard & 0.900 & 3.346 & 0.152 \\
    Prior baseline & Standard & 0.910 & 3.740 & 0.138 \\
    \bottomrule
  \end{tabular}
\end{table}

\begin{figure}[t]
  \centering
  \includegraphics[width=\textwidth]{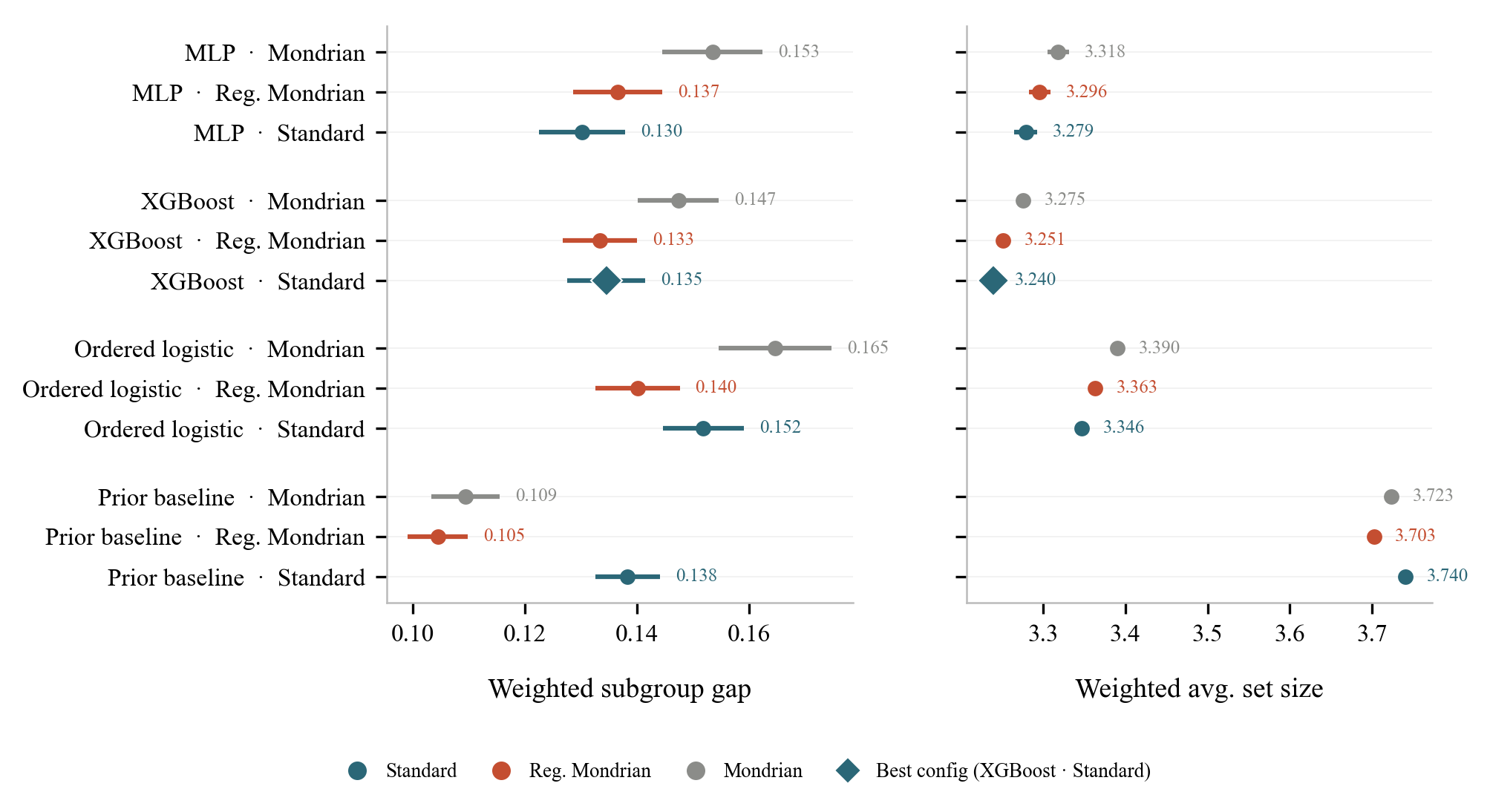}
  \caption{\textbf{Fairness-efficiency trade-off.} All model$\times$method combinations (100-split weighted means with 95\% CIs). \textbf{Left}: weighted subgroup gap (lower is fairer). \textbf{Right}: weighted average set size (lower is more efficient). For XGBoost, Mondrian (gray) increases both gap and set size relative to standard (teal); regularized Mondrian (red) reduces the gap penalty. Diamond marks the XGBoost-standard reference configuration.}
  \label{fig:tradeoff}
\end{figure}

\subsection{Per-group coverage and failure analysis}
\label{sec:res_pergroup}

Figure~\ref{fig:subgroup} decomposes XGBoost coverage by subgroup. Under standard conformal, coverage ranges from 0.867 (Black non-Hispanic $|$ College+, $n{=}151$) to 0.923 (Asian/Other $|$ Some college, $n{=}144$), a 5.6~percentage-point spread. This quantity is the range of \emph{groupwise mean} coverages. By contrast, the 0.135 value in Table~\ref{tab:main} is the \emph{mean of splitwise max$-$min gaps}; these summaries need not match numerically. The three most under-covered groups are Black non-Hispanic $|$ College+ (0.867), Hispanic $|$ College+ (0.887), and Hispanic $|$ HS or less (0.888), all falling below the 90\% target. Mondrian overcorrects for some under-covered groups (e.g., raising Black non-Hispanic $|$ College+ to 0.925) but inflates coverage in already-overcovered groups, widening the overall gap. Table~\ref{tab:pergroup} (Appendix~\ref{app:pergroup}) reports all 12 per-group coverage values.

Figure~\ref{fig:failure} traces the mechanism behind Mondrian's negative result. The heatmap shows per-group deltas (Mondrian minus standard) for coverage, set size, and threshold, sorted by calibration-cell size. The smallest cells experience the largest perturbations: Black non-Hispanic $|$ College+ ($n_\mathrm{cal}{=}45$) sees a set-size increase of $+0.442$ labels under vanilla Mondrian, while the largest cell (White non-Hispanic $|$ HS or less, $n_\mathrm{cal}{=}355$) shifts by only $-0.033$. Regularized Mondrian (right panel) compresses these perturbations: the maximum set-size delta drops from $+0.442$ to $+0.159$. The underlying mechanism is calibration-cell fragmentation interacting with group-specific confidence mismatch. In thin cells, the group-specific threshold $\hat{q}_{g,\alpha}$ is estimated from few observations, causing it to overreact to local score-distribution idiosyncrasies. When the base predictor's confidence pattern varies across groups, as it does for XGBoost which fits group-varying feature interactions, this overreaction amplifies rather than attenuates coverage disparities.

\begin{figure}[t]
  \centering
  \includegraphics[width=\textwidth]{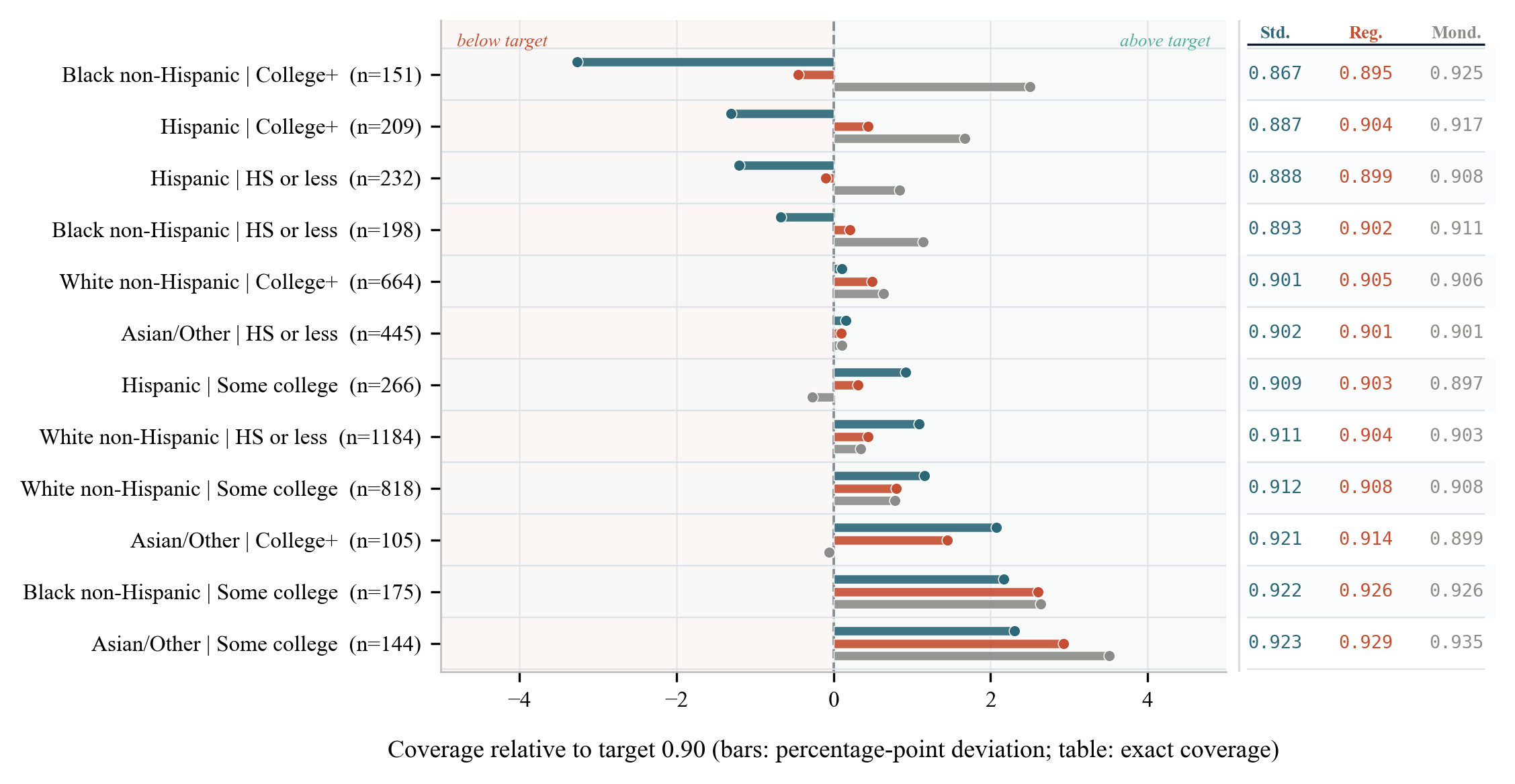}
  \caption{\textbf{XGBoost per-group weighted coverage.} Coverage relative to 90\% target across 12 race$\times$education cells (100-split means). Bars show percentage-point deviation from target; right table shows absolute coverage. Groups ordered by standard coverage. Standard (teal) under-covers minority $|$ College+ groups; Mondrian (gray) overcorrects, widening the spread.}
  \label{fig:subgroup}
\end{figure}

\begin{figure}[t]
  \centering
  \includegraphics[width=\textwidth]{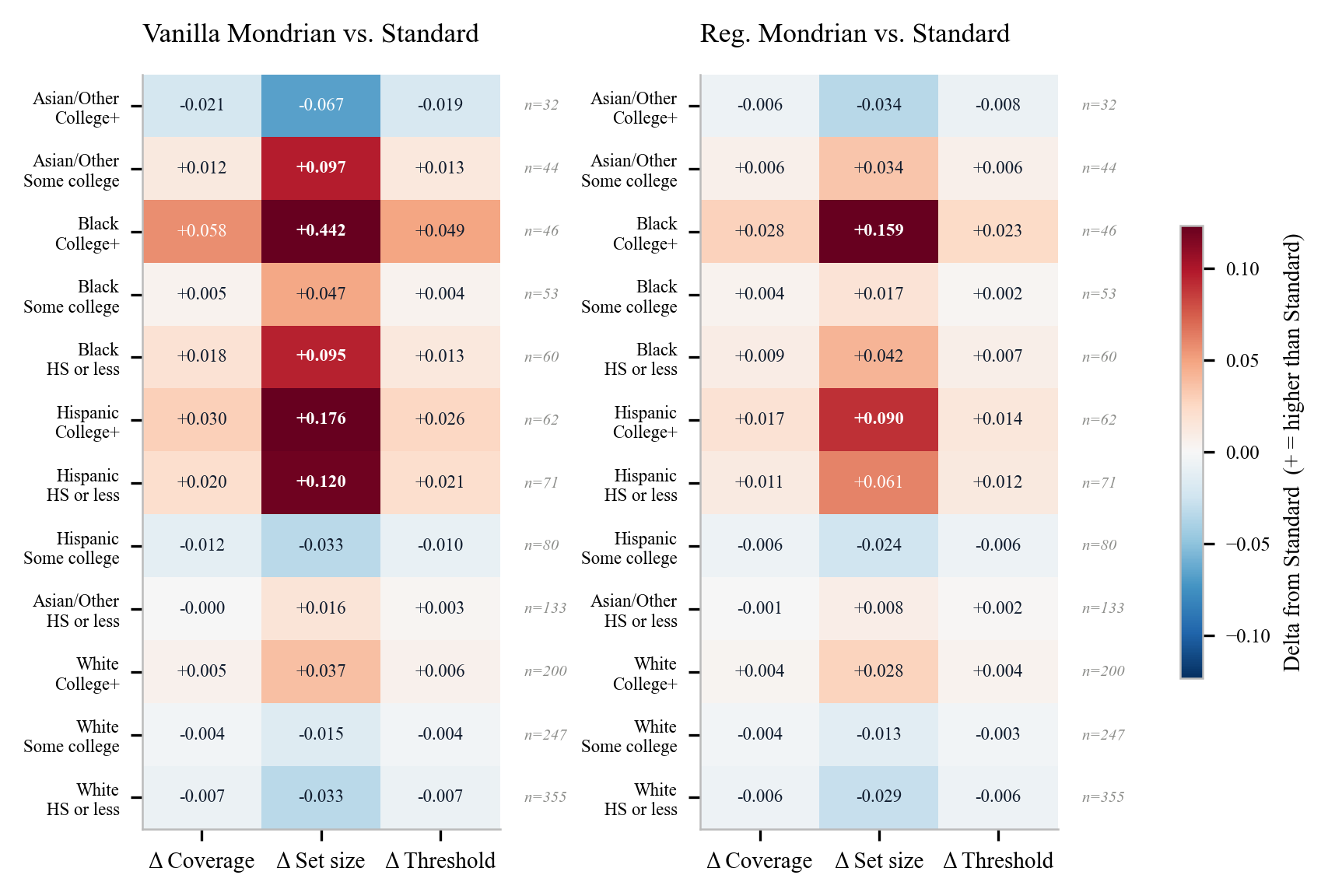}
  \caption{\textbf{Per-group deltas for XGBoost.} Method minus standard, sorted by calibration-cell size ($n_\mathrm{cal}$). \textbf{Left}: Mondrian. \textbf{Right}: regularized Mondrian. Thin cells (top rows) experience the largest threshold overreaction; shrinkage compresses perturbations toward zero.}
  \label{fig:failure}
\end{figure}

\subsection{Survey weights and predictor efficiency}
\label{sec:res_weights}

Switching from unweighted to survey-weighted evaluation increases the subgroup gap by approximately $+0.020$ for XGBoost under standard conformal (from 0.115 unweighted to 0.135 weighted), consistent across all methods (see Appendix~\ref{app:weights}). Coverage shifts are negligible ($+0.002$), and weighted set sizes are slightly smaller ($-0.025$). The effect arises because design-based weights up-weight underrepresented groups whose coverage deviates from the mean. While modest in magnitude, this consistent upward shift confirms that unweighted evaluation systematically understates subgroup disparities in complex surveys.

We also tested survey-weighted thresholds. They perform poorly: for XGBoost, weighted Mondrian drops weighted coverage to 0.898 and widens the subgroup gap to 0.177, worse than both standard (0.135) and unweighted Mondrian (0.147); the same pattern holds across all four predictors (Table~\ref{tab:full_results}, Appendix~\ref{app:allresults}). In our setting, weights matter for \emph{evaluation}, but naive weight injection into calibration is not a reliable remedy.

All predictors achieve comparable marginal coverage (Table~\ref{tab:main}), but efficiency varies substantially. XGBoost produces the smallest prediction sets (3.240 labels on average out of 5), followed by \acs{mlp} (3.279), ordered logistic (3.346), and the prior baseline (3.740). The prior baseline's prediction sets are approximately 15\% wider than XGBoost's while achieving similar coverage, confirming that predictor quality translates directly into conformal efficiency: better-calibrated class probabilities yield tighter sets without sacrificing validity.

\subsection{Sensitivity analyses confirm robustness}
\label{sec:res_sensitivity}

Figure~\ref{fig:sensitivity} shows that the central negative result persists across outcome operationalizations. When ``Not sure'' responses are merged into the ordinal midpoint ($n{=}5{,}348$), weighted gaps are somewhat smaller (standard: 0.124, Mondrian: 0.143) but Mondrian still widens the gap. Under the binary trust item ($n{=}4{,}071$), all gaps are larger (standard: 0.156, Mondrian: 0.187) and the Mondrian penalty is more pronounced ($\Delta = {+}0.031$). In both branches, regularized Mondrian falls between standard and Mondrian, confirming its stabilizing role. Alternate subgroup families (race-only with 4 groups, education-only with 3 groups) yield qualitatively identical conclusions.

\begin{figure}[t]
  \centering
  \includegraphics[width=0.85\textwidth]{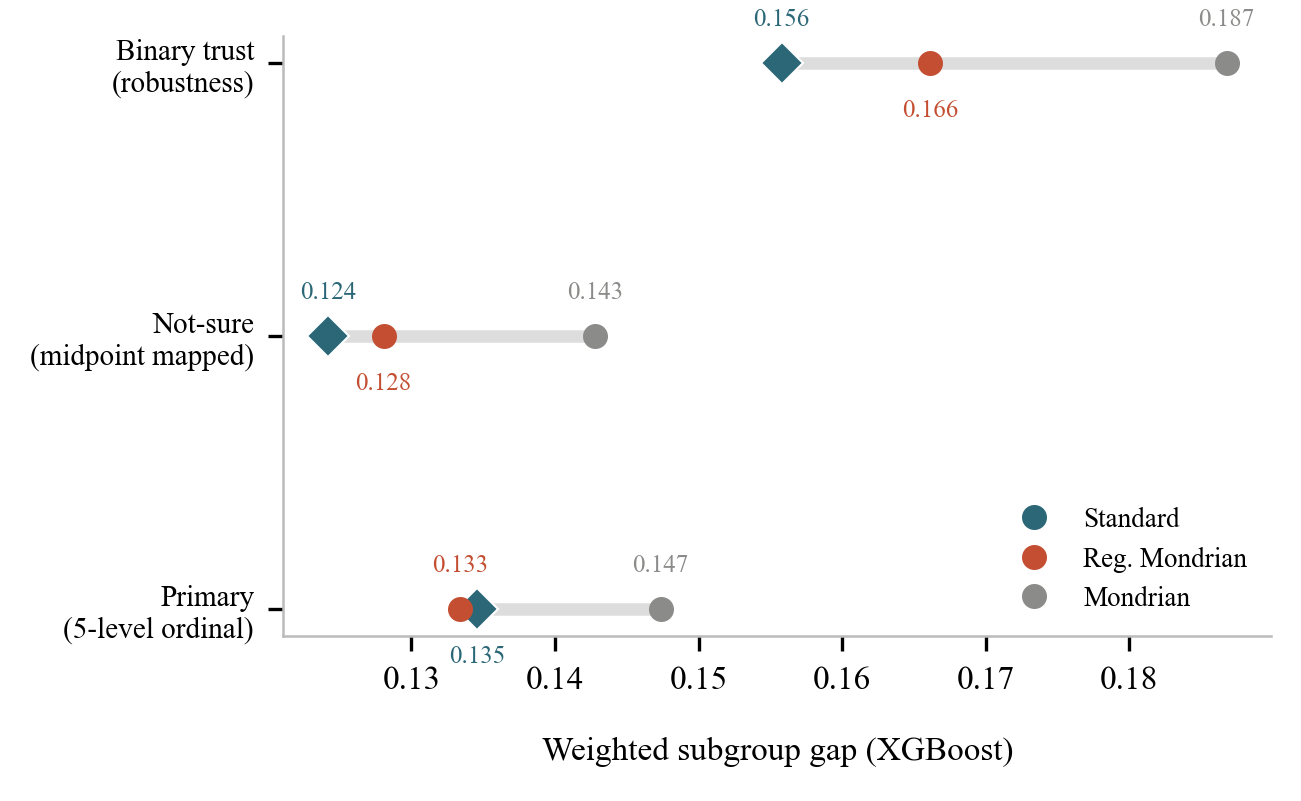}
  \caption{\textbf{Weighted subgroup gap across outcome operationalizations.} XGBoost results (100-split means). Mondrian (gray) widens the gap relative to standard (teal diamond) in all branches. Regularized Mondrian (red) consistently falls between the two.}
  \label{fig:sensitivity}
\end{figure}

\subsection{Predictor informativeness moderates Mondrian failure}
\label{sec:res_mechanism}

Cross-predictor results in Table~\ref{tab:paired_all} (Appendix~\ref{app:paired}) identify predictor informativeness as a key moderator. For the three informative predictors, Mondrian increases both subgroup gap and set size: XGBoost ($\Delta_{\mathrm{gap}}{=}{+}0.013$, $\Delta_{\mathrm{size}}{=}{+}0.036$), \acs{mlp} ($+0.023$, $+0.039$), and ordered logistic ($+0.013$, $+0.043$); effect sizes on set size are large ($d_z{=}1.66$, $1.64$, $2.06$). The prior baseline inverts this pattern, with Mondrian reducing both gap ($-0.029$) and set size ($-0.017$). Informative predictors induce heterogeneous score distributions across subgroups, so thin-cell per-group thresholds overreact to local idiosyncrasies and amplify disparities; with the uninformative prior baseline, scores are nearly homogeneous and thresholds stay close to the global value. Mondrian safety therefore requires auditing both group-varying confidence patterns and calibration-cell sizes.

\section{Discussion}
\label{sec:discussion}

Our central finding is that Mondrian \acs{cp} can worsen the fairness-efficiency trade-off. For XGBoost, Mondrian increases both weighted subgroup gap ($+0.013$) and weighted set size ($+0.036$) relative to standard conformal; the same adverse direction appears for \acs{mlp}, ordered logistic, and the sensitivity branches. A key conceptual point is that Mondrian's guarantee target and our audit target are not identical. Mondrian controls within-group unweighted conditional coverage for predefined groups under exchangeability, whereas our deployment-relevant criterion is a survey-weighted max$-$min coverage gap across intersectional cells over repeated random splits. In thin, imbalanced cells, satisfying the former need not improve the latter. This mismatch helps explain why a theorem-backed group-specific calibration rule can still look worse under fairness metrics that matter in practice.

The failure is not universal: Section~\ref{sec:res_mechanism} shows that the prior baseline is a boundary case where nearly homogeneous score distributions let Mondrian reduce the weighted gap (0.138 to 0.109). But thin calibration cells still make group thresholds unstable: the thinnest cell (Black non-Hispanic $|$ College+, $n_\mathrm{cal}{=}45$) shows a set-size increase of $+0.442$ under Mondrian, versus $-0.033$ for the largest cell~\citep{lei2014distribution, foygel2021limits}.

Regularized Mondrian ($\lambda{=}50$) partially mitigates this instability by shrinking thin-cell thresholds toward the global value. For XGBoost, shrinkage reduces the maximum set-size perturbation ($+0.442 \rightarrow +0.159$) and yields the best weighted gap among XGBoost variants (0.133 vs.\ 0.135 standard, 0.147 Mondrian). Survey weighting also matters: weighted evaluation increases the XGBoost subgroup gap by about $+0.020$ (Appendix~\ref{app:weights}), so unweighted analyses understate disparity in complex surveys~\citep{lumley2010complex}. Mondrian should therefore not be assumed to improve subgroup fairness in thin-cell settings; the calibration guarantee, audited subgroup structure, and weighting scheme must be aligned in deployment, with per-group diagnostics and a shrinkage comparator as minimum safeguards~\citep{cresswell2024conformal, romano2020malice}.

\section{Conclusion}
\label{sec:conclusion}

We have shown that standard split \acs{cp} achieves nominal marginal coverage while masking subgroup disparities ({$\sim$}13~pp weighted gap across 12 race$\times$education cells). Mondrian \acs{cp} worsens both gap and efficiency for informative predictors when calibration cells are thin ($d_z{=}1.66$ on set size, $\Delta_{\mathrm{gap}} = +0.013$ for XGBoost), and a James-Stein-style regularized comparator only partly mitigates this instability. Failure analysis pins the mechanism on calibration-cell fragmentation interacting with group-specific confidence mismatch, a condition that is absent for uninformative predictors and thus helps explain why Mondrian helps the prior baseline but harms XGBoost, \acs{mlp}, and ordered logistic. Survey-weighted evaluation consistently reveals larger disparities than unweighted analysis ($+0.020$ for XGBoost), reinforcing that design weights belong in the evaluation pipeline even when they should not be injected naively into calibration. Together, these results argue that marginal coverage certificates are insufficient for equitable deployment, and that per-group diagnostics, cell-size checks, and a shrinkage comparator should be standard pre-deployment practice.

Results derive from a single U.S. survey wave, subgroup boundaries are fixed a priori, the regularized comparator lacks formal guarantees, and AI-attitude correlates are same-wave rather than strictly exogenous. Open directions include adaptive or data-driven subgroup selection, cross-wave validation under temporal distribution shift, design-aware conformal theory that formally incorporates survey weights into calibration, and extensions to regression or multi-label ordinal outcomes.


\bibliographystyle{plainnat}
\bibliography{references}

\appendix

\section{Coverage guarantee proofs}
\label{app:proofs}

We restate the finite-sample coverage guarantees invoked in Section~\ref{sec:conformal} for completeness. These results are standard; we include them to make the paper self-contained.

\subsection{Standard split conformal coverage}
\label{app:proof_standard}

\begin{proposition}[Marginal coverage, \citealt{vovk2005algorithmic}]
Let $\{(X_i, Y_i)\}_{i=1}^{n+1}$ be exchangeable random variables. Define nonconformity scores $S_i = s(X_i, Y_i)$ for $i \in I_\mathrm{cal}$ and the threshold $\hat{q}_\alpha$ as in Eq.~\eqref{eq:standard_threshold}. Then
\begin{equation}
  \Pr\!\bigl(Y_{n+1} \in \hat{C}_\alpha(X_{n+1})\bigr) \ge 1 - \alpha.
  \label{eq:proof_marginal}
\end{equation}
\end{proposition}

\begin{proof}
Since the calibration scores $\{S_i\}_{i \in I_\mathrm{cal}}$ and the test score $S_{n+1}$ are exchangeable, the rank of $S_{n+1}$ among $\{S_1, \dots, S_{|I_\mathrm{cal}|}, S_{n+1}\}$ is uniformly distributed over $\{1, \dots, |I_\mathrm{cal}| + 1\}$. The threshold $\hat{q}_\alpha = \mathrm{Quantile}_{1-\alpha}(\{S_i\} \cup \{\infty\})$ is the $\lceil (1-\alpha)(|I_\mathrm{cal}|+1) \rceil$-th smallest value. Therefore,
\begin{equation}
  \Pr(S_{n+1} \le \hat{q}_\alpha) = \frac{\lceil (1-\alpha)(|I_\mathrm{cal}|+1) \rceil}{|I_\mathrm{cal}|+1} \ge 1 - \alpha.
\end{equation}
Since $Y_{n+1} \in \hat{C}_\alpha(X_{n+1})$ if and only if $S_{n+1} \le \hat{q}_\alpha$, the result follows.
\end{proof}

\subsection{Mondrian conditional coverage}
\label{app:proof_mondrian}

\begin{proposition}[Group-conditional coverage, \citealt{vovk2005algorithmic, vovk2012conditional}]
Under within-group exchangeability of calibration and test data, the Mondrian prediction set with threshold $\hat{q}_{g,\alpha}$ as in Eq.~\eqref{eq:mondrian_threshold} satisfies
\begin{equation}
  \Pr\!\bigl(Y_{n+1} \in \hat{C}^\mathrm{Mond}_\alpha(X_{n+1}) \mid G_{n+1} = g\bigr) \ge 1 - \alpha \quad \text{for each } g \in \{1, \dots, G\}.
  \label{eq:proof_conditional}
\end{equation}
\end{proposition}

\begin{proof}
Conditioning on $G_{n+1} = g$, the scores $\{S_i : i \in I_\mathrm{cal}^{(g)}\} \cup \{S_{n+1}\}$ form an exchangeable set of size $n_g + 1$. By the same rank argument as Proposition~1 applied within group $g$:
\begin{equation}
  \Pr(S_{n+1} \le \hat{q}_{g,\alpha} \mid G_{n+1} = g) = \frac{\lceil (1-\alpha)(n_g+1) \rceil}{n_g+1} \ge 1 - \alpha.
\end{equation}
The key observation for our study is that the finite-sample correction $1/(n_g+1)$ grows with $1/n_g$. For the thinnest cell in our data ($n_g = 32$), the effective coverage floor is $1 - \alpha + 1/(n_g+1) \approx 0.930$, whereas for the largest cell ($n_g = 355$) it is $\approx 0.903$. This differential tightness contributes to the threshold variance documented in Section~\ref{sec:res_pergroup}.
\end{proof}

\section{Regularized Mondrian derivation}
\label{app:shrinkage}

The regularized Mondrian threshold (Eq.~\eqref{eq:reg_threshold}) shrinks the group-specific quantile toward the global quantile:
\begin{equation}
  \hat{q}_{g,\alpha}^\mathrm{reg} = w_g \cdot \hat{q}_{g,\alpha} + (1 - w_g) \cdot \hat{q}_\alpha, \qquad w_g = \frac{n_g}{n_g + \lambda}.
\end{equation}

This construction follows the James-Stein shrinkage principle~\citep{james1961estimation, efron1975data}. We briefly derive the motivation and state properties.

\textbf{MSE motivation.} Consider the group-specific threshold $\hat{q}_{g,\alpha}$ as a noisy estimate of the oracle group quantile $q_g^*$. Assume $\mathrm{Var}(\hat{q}_{g,\alpha}) \approx \sigma^2 / n_g$ for some constant $\sigma^2$ (the quantile variance scales inversely with sample size by the Bahadur representation). The global threshold $\hat{q}_\alpha$ has variance $\sigma^2 / n_\mathrm{cal}$. A convex combination $\hat{q}^\mathrm{reg} = w \hat{q}_g + (1-w) \hat{q}_\alpha$ has MSE:
\begin{equation}
  \mathrm{MSE}(w) = w^2 \cdot \frac{\sigma^2}{n_g} + (1-w)^2 \cdot \left(\frac{\sigma^2}{n_\mathrm{cal}} + b_g^2\right),
  \label{eq:mse_shrinkage}
\end{equation}
where $b_g = q_g^* - q^*$ is the bias from using the global quantile. Minimizing over $w$ yields:
\begin{equation}
  w^* = \frac{\sigma^2/n_\mathrm{cal} + b_g^2}{\sigma^2/n_g + \sigma^2/n_\mathrm{cal} + b_g^2}.
\end{equation}
When $b_g \approx 0$ (group quantile close to global) and $n_\mathrm{cal} \gg n_g$, this simplifies to $w^* \approx n_g / (n_g + \text{const})$. Setting $\lambda = 50$ yields a conservative plug-in that matches the scale of our thinnest cells ($n_g = 32\text{--}45$).

\textbf{Shrinkage weights.} Table~\ref{tab:shrinkage} reports the shrinkage weights for all 12 groups under $\lambda = 50$.

\begin{table}[h]
  \caption{\textbf{Shrinkage weights.} $w_g = n_g / (n_g + 50)$ for $\lambda = 50$. Groups sorted by calibration-cell size.}
  \label{tab:shrinkage}
  \centering
  \small
  \begin{tabular}{lrr}
    \toprule
    Group & $n_\mathrm{cal}$ & $w_g$ \\
    \midrule
    Asian/Other $|$ College+     & 32  & 0.39 \\
    Asian/Other $|$ Some college & 43  & 0.46 \\
    Black non-Hisp.\ $|$ College+ & 45  & 0.47 \\
    Black non-Hisp.\ $|$ Some college & 53  & 0.51 \\
    Black non-Hisp.\ $|$ HS or less & 59  & 0.54 \\
    Hispanic $|$ College+        & 63  & 0.56 \\
    Hispanic $|$ HS or less      & 70  & 0.58 \\
    Hispanic $|$ Some college    & 80  & 0.62 \\
    Asian/Other $|$ HS or less   & 134 & 0.73 \\
    White non-Hisp.\ $|$ College+ & 199 & 0.80 \\
    White non-Hisp.\ $|$ Some college & 245 & 0.83 \\
    White non-Hisp.\ $|$ HS or less & 355 & 0.88 \\
    \bottomrule
  \end{tabular}
\end{table}

\textbf{Coverage note.} The regularized threshold does not carry a formal finite-sample coverage guarantee because the convex combination breaks the rank-based argument. Empirically, marginal coverage remains within 1~pp of the nominal target across all 100 splits (Table~\ref{tab:main}), suggesting that shrinkage does not meaningfully degrade aggregate validity.

\section{Base predictor details}
\label{app:base}

Table~\ref{tab:base_metrics} reports test-set performance for all four base predictors across 100 splits. XGBoost achieves the highest accuracy and lowest log-loss, confirming its role as the strongest practical predictor.

\begin{table}[h]
  \caption{\textbf{Base predictor test-set performance.} 100-split means $\pm$ SD. Accuracy and log-loss are unweighted; weighted variants show similar rankings.}
  \label{tab:base_metrics}
  \centering
  \small
  \begin{tabular}{lcccc}
    \toprule
    Model & Features & Accuracy & Wtd.\ Accuracy & Log-loss \\
    \midrule
    XGBoost          & 66 & $0.525 \pm 0.015$ & $0.519 \pm 0.016$ & $1.098 \pm 0.023$ \\
    \acs{mlp}        & 66 & $0.510 \pm 0.019$ & $0.505 \pm 0.019$ & $1.117 \pm 0.041$ \\
    Ordered logistic & 14 & $0.446 \pm 0.014$ & $0.438 \pm 0.015$ & $1.244 \pm 0.018$ \\
    Prior baseline   &  0 & $0.393 \pm 0.000$ & $0.393 \pm 0.000$ & $1.456 \pm 0.000$ \\
    \bottomrule
  \end{tabular}
\end{table}

\textbf{Model specifications.} The \textbf{prior baseline} assigns the training-set marginal class distribution to all test points ($\hat{p}_c = n_c^\mathrm{tr} / |I_\mathrm{tr}|$), providing a lower bound on predictor quality. \textbf{Ordered logistic} uses a proportional-odds model with 14 features: age, gender, education, race/ethnicity, income, party ID, ideology, religious attendance, region, urbanicity, internet use, smartphone ownership, social media use, and AI familiarity. \textbf{XGBoost} uses 100 boosting rounds, max depth 4, learning rate 0.1, and $\ell_2$ regularization $\lambda = 1.0$ with softmax objective. The \textbf{\acs{mlp}} has two hidden layers (128, 64 units), ReLU activations, dropout 0.3, and is trained for up to 100 epochs with early stopping (patience 10) using Adam ($\mathrm{lr}{=}10^{-3}$). All models converge across all 100 splits.

\section{Full experimental results}
\label{app:full}

\subsection{All model-method combinations}
\label{app:allresults}

Table~\ref{tab:full_results} extends Table~\ref{tab:main} to all 16 model$\times$method combinations evaluated in the primary branch (100-split survey-weighted means).

\begin{table}[h]
  \caption{\textbf{Full results for all model-method combinations.} 100-split survey-weighted means ($\alpha = 0.10$).}
  \label{tab:full_results}
  \centering
  \small
  \begin{tabular}{llccc}
    \toprule
    Model & Method & Wtd.\ Coverage & Wtd.\ Set Size & Wtd.\ Gap \\
    \midrule
    \multirow{4}{*}{XGBoost}
      & Standard     & 0.905 & 3.240 & 0.135 \\
      & Reg.\ Mond.\ & 0.906 & 3.251 & 0.133 \\
      & Mondrian     & 0.908 & 3.275 & 0.147 \\
      & Wtd.\ Mond.\ & 0.898 & 3.223 & 0.177 \\
    \midrule
    \multirow{4}{*}{\acs{mlp}}
      & Standard     & 0.904 & 3.279 & 0.130 \\
      & Reg.\ Mond.\ & 0.906 & 3.296 & 0.137 \\
      & Mondrian     & 0.908 & 3.318 & 0.153 \\
      & Wtd.\ Mond.\ & 0.898 & 3.264 & 0.178 \\
    \midrule
    \multirow{4}{*}{Ordered logistic}
      & Standard     & 0.900 & 3.346 & 0.152 \\
      & Reg.\ Mond.\ & 0.903 & 3.363 & 0.140 \\
      & Mondrian     & 0.906 & 3.390 & 0.165 \\
      & Wtd.\ Mond.\ & 0.897 & 3.348 & 0.198 \\
    \midrule
    \multirow{4}{*}{Prior baseline}
      & Standard     & 0.910 & 3.740 & 0.138 \\
      & Reg.\ Mond.\ & 0.908 & 3.703 & 0.105 \\
      & Mondrian     & 0.911 & 3.723 & 0.109 \\
      & Wtd.\ Mond.\ & 0.899 & 3.649 & 0.145 \\
    \bottomrule
  \end{tabular}
\end{table}

\subsection{Survey-weight shift analysis}
\label{app:weights}

Figure~\ref{fig:weight_shift} shows the distribution of weighted$-$unweighted metric differences across 100 splits for XGBoost. Survey weighting consistently increases the subgroup gap (mean shift: $+0.020$ for standard, $+0.016$ for regularized Mondrian, $+0.010$ for Mondrian) and slightly decreases set sizes ($-0.025$ for standard). Coverage shifts are negligible ($<0.003$). Table~\ref{tab:weight_shift} reports means and standard deviations for XGBoost methods.

\begin{figure}[h]
  \centering
  \includegraphics[width=\textwidth]{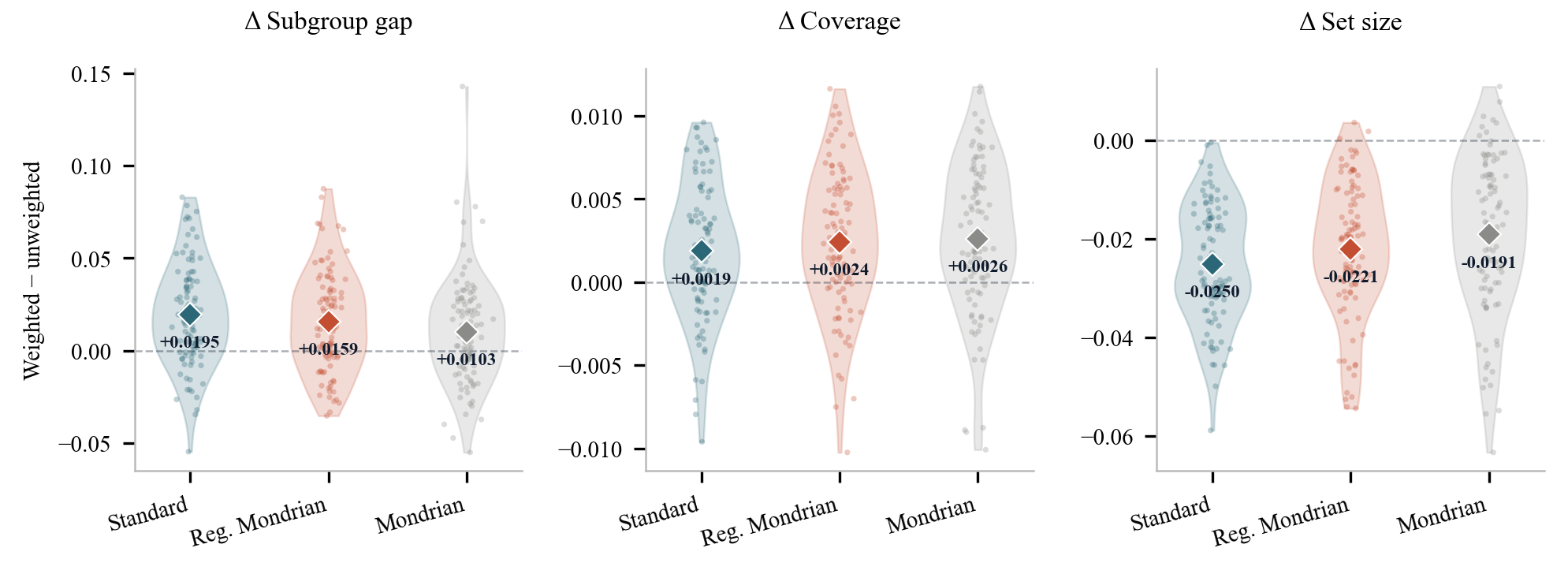}
  \caption{\textbf{Survey-weighting shifts for XGBoost.} Weighted$-$unweighted differences across 100 splits. \textbf{Left}: subgroup gap increases under weighting. \textbf{Center}: coverage is negligibly affected. \textbf{Right}: weighted set sizes are slightly smaller.}
  \label{fig:weight_shift}
\end{figure}

\begin{table}[h]
  \caption{\textbf{Survey-weight shift for XGBoost.} Weighted $-$ unweighted (100-split means $\pm$ SD).}
  \label{tab:weight_shift}
  \centering
  \small
  \begin{tabular}{lccc}
    \toprule
    Method & $\Delta$ Coverage & $\Delta$ Gap & $\Delta$ Set Size \\
    \midrule
    Standard       & $+0.002 \pm 0.004$ & $+0.020 \pm 0.027$ & $-0.025 \pm 0.012$ \\
    Reg.\ Mond.\   & $+0.002 \pm 0.004$ & $+0.016 \pm 0.027$ & $-0.022 \pm 0.014$ \\
    Mondrian       & $+0.003 \pm 0.004$ & $+0.010 \pm 0.029$ & $-0.019 \pm 0.016$ \\
    \bottomrule
  \end{tabular}
\end{table}

\subsection{Full per-group coverage}
\label{app:pergroup}

Table~\ref{tab:pergroup} reports XGBoost per-group weighted coverage for all 12 race$\times$education cells under the three conformal methods (100-split means).

\begin{table}[h]
  \caption{\textbf{XGBoost per-group weighted coverage.} 100-split means. Groups sorted by standard coverage.}
  \label{tab:pergroup}
  \centering
  \small
  \begin{tabular}{lcccr}
    \toprule
    Group & Standard & Reg.\ Mond.\ & Mondrian & $n$ \\
    \midrule
    Black non-Hisp.\ $|$ College+ & 0.867 & 0.895 & 0.925 & 151 \\
    Hispanic $|$ College+ & 0.887 & 0.904 & 0.917 & 209 \\
    Hispanic $|$ HS or less & 0.888 & 0.899 & 0.908 & 232 \\
    Black non-Hisp.\ $|$ HS or less & 0.893 & 0.902 & 0.911 & 198 \\
    White non-Hisp.\ $|$ College+ & 0.901 & 0.905 & 0.906 & 664 \\
    Asian/Other $|$ HS or less & 0.902 & 0.901 & 0.901 & 445 \\
    Hispanic $|$ Some college & 0.909 & 0.903 & 0.897 & 266 \\
    White non-Hisp.\ $|$ HS or less & 0.911 & 0.904 & 0.903 & 1184 \\
    White non-Hisp.\ $|$ Some college & 0.912 & 0.908 & 0.908 & 818 \\
    Asian/Other $|$ College+ & 0.921 & 0.914 & 0.899 & 105 \\
    Black non-Hisp.\ $|$ Some college & 0.922 & 0.926 & 0.926 & 175 \\
    Asian/Other $|$ Some college & 0.923 & 0.929 & 0.935 & 144 \\
    \bottomrule
  \end{tabular}
\end{table}

\subsection{Paired method comparisons}
\label{app:paired}

Table~\ref{tab:paired} reports paired deltas (method $-$ standard) for XGBoost across 100 splits with 95\% confidence intervals and Cohen's $d_z$.

\begin{table}[h]
  \caption{\textbf{Paired deltas for XGBoost.} Method $-$ standard, 100 splits. CI = mean $\pm$ 1.96$\cdot$SE.}
  \label{tab:paired}
  \centering
  \small
  \begin{tabular}{lcccc}
    \toprule
    Comparison & $\Delta$ Wtd.\ Gap [95\% CI] & $d_z$ (Gap) & $\Delta$ Wtd.\ Size [95\% CI] & $d_z$ (Size) \\
    \midrule
    Mondrian $-$ Std.\ & ${+}0.013\;[{+}0.004, {+}0.021]$ & 0.30 & ${+}0.036\;[{+}0.032, {+}0.040]$ & 1.66 \\
    Reg.\ Mond.\ $-$ Std.\ & ${-}0.001\;[{-}0.007, {+}0.004]$ & 0.04 & ${+}0.012\;[{+}0.009, {+}0.015]$ & 0.75 \\
    \bottomrule
  \end{tabular}
\end{table}

Table~\ref{tab:paired_all} extends the paired comparison to all four base predictors.

\begin{table}[h]
  \caption{\textbf{Paired deltas across all predictors.} Method $-$ standard, 100 splits. Coverage, gap, and set-size deltas with Cohen's $d_z$.}
  \label{tab:paired_all}
  \centering
  \small
  \begin{tabular}{llcccc}
    \toprule
    Model & Comparison & $\Delta$ Cov.\ & $\Delta$ Gap & $\Delta$ Size & $d_z$ (Size) \\
    \midrule
    \multirow{2}{*}{XGBoost}
      & Mondrian     & ${+}0.004$ & ${+}0.013$ & ${+}0.036$ & 1.66 \\
      & Reg.\ Mond.\ & ${+}0.001$ & ${-}0.001$ & ${+}0.012$ & 0.75 \\
    \midrule
    \multirow{2}{*}{\acs{mlp}}
      & Mondrian     & ${+}0.004$ & ${+}0.023$ & ${+}0.039$ & 1.64 \\
      & Reg.\ Mond.\ & ${+}0.002$ & ${+}0.006$ & ${+}0.017$ & 0.92 \\
    \midrule
    \multirow{2}{*}{Ord.\ logistic}
      & Mondrian     & ${+}0.006$ & ${+}0.013$ & ${+}0.043$ & 2.06 \\
      & Reg.\ Mond.\ & ${+}0.004$ & ${-}0.012$ & ${+}0.017$ & 1.15 \\
    \midrule
    \multirow{2}{*}{Prior baseline}
      & Mondrian     & ${+}0.001$ & ${-}0.029$ & ${-}0.017$ & $-$0.50 \\
      & Reg.\ Mond.\ & ${-}0.002$ & ${-}0.034$ & ${-}0.038$ & $-$1.45 \\
    \bottomrule
  \end{tabular}
\end{table}

\section{Failure analysis details}
\label{app:failure}

\subsection{Group-specific confidence mismatch}
\label{app:mismatch}

Table~\ref{tab:mismatch} reports the per-group overconfidence metric for XGBoost, defined as $\text{overconfidence}_g = \bar{p}_{\max,g} - \text{accuracy}_g$, where $\bar{p}_{\max,g}$ is the mean predicted probability for the most-likely class. Positive values indicate the model is more confident than accurate. The variation across groups (range: 0.017 to 0.116) drives Mondrian threshold heterogeneity.

\begin{table}[h]
  \caption{\textbf{XGBoost confidence mismatch by group.} 100-split weighted means. Overconfidence $= \bar{p}_{\max} - \text{accuracy}$. Groups sorted by overconfidence.}
  \label{tab:mismatch}
  \centering
  \small
  \begin{tabular}{lccc}
    \toprule
    Group & Wtd.\ Accuracy & Wtd.\ Confidence & Overconfidence \\
    \midrule
    Asian/Other $|$ Some college          & 0.544 & 0.561 & 0.017 \\
    White non-Hisp.\ $|$ Some college     & 0.548 & 0.582 & 0.034 \\
    White non-Hisp.\ $|$ HS or less       & 0.535 & 0.570 & 0.035 \\
    Asian/Other $|$ College+              & 0.523 & 0.566 & 0.043 \\
    Hispanic $|$ HS or less               & 0.517 & 0.565 & 0.049 \\
    Black non-Hisp.\ $|$ College+         & 0.515 & 0.569 & 0.053 \\
    White non-Hisp.\ $|$ College+         & 0.502 & 0.564 & 0.062 \\
    Asian/Other $|$ HS or less            & 0.490 & 0.561 & 0.071 \\
    Black non-Hisp.\ $|$ Some college     & 0.482 & 0.560 & 0.078 \\
    Black non-Hisp.\ $|$ HS or less       & 0.472 & 0.565 & 0.093 \\
    Hispanic $|$ Some college             & 0.536 & 0.562 & 0.026 \\
    Hispanic $|$ College+                 & 0.443 & 0.559 & 0.116 \\
    \bottomrule
  \end{tabular}
\end{table}

The group with the largest overconfidence (Hispanic $|$ College+, 0.116) is also among the three most under-covered under standard conformal (coverage 0.887, Table~\ref{tab:pergroup}), confirming that confidence mismatch translates into coverage shortfalls.

\subsection{Failure extrema concentration}
\label{app:extrema}

Table~\ref{tab:extrema} reports how often each group is the single worst-affected cell under Mondrian for XGBoost across 100 splits, for two failure modes: maximum set-size inflation and maximum coverage-error worsening.

\begin{table}[h]
  \caption{\textbf{Failure extrema concentration under XGBoost Mondrian.} Frequency (out of 100 splits) that each group is the single worst-affected cell. Top 5 groups shown per failure mode.}
  \label{tab:extrema}
  \centering
  \small
  \begin{tabular}{lclc}
    \toprule
    \multicolumn{2}{c}{Max set-size inflation} & \multicolumn{2}{c}{Max coverage-error worsening} \\
    \cmidrule(lr){1-2} \cmidrule(lr){3-4}
    Group & Freq.\ & Group & Freq.\ \\
    \midrule
    Black non-Hisp.\ $|$ College+   & 66 & Black non-Hisp.\ $|$ College+   & 18 \\
    Hispanic $|$ College+            & 14 & Asian/Other $|$ College+         & 12 \\
    Asian/Other $|$ Some college     &  6 & Black non-Hisp.\ $|$ HS or less & 11 \\
    Black non-Hisp.\ $|$ HS or less  &  5 & Black non-Hisp.\ $|$ Some college & 11 \\
    Hispanic $|$ HS or less          &  5 & Hispanic $|$ College+            & 11 \\
    \bottomrule
  \end{tabular}
\end{table}

Black non-Hispanic $|$ College+ ($n_\mathrm{cal}{=}45$) is the worst set-size cell in 66 out of 100 splits, confirming that failure concentrates in the thinnest minority cells rather than distributing randomly.

\subsection{Cell-size correlation with Mondrian perturbations}
\label{app:correlation}

Table~\ref{tab:corr} reports Pearson and Spearman correlations between calibration-cell size ($n_g$) and Mondrian-induced perturbations across all 12 groups $\times$ 100 splits ($n{=}1{,}200$ observations per model).

\begin{table}[h]
  \caption{\textbf{Cell-size correlation with Mondrian perturbation.} 1,200 group$\times$split observations. All $p < 0.001$ unless noted.}
  \label{tab:corr}
  \centering
  \small
  \begin{tabular}{llcc}
    \toprule
    Model & Metric & Pearson $r$ & Spearman $\rho$ \\
    \midrule
    \multirow{2}{*}{XGBoost}
      & $\Delta$ Set size   & $-0.247$ & $-0.225$ \\
      & $\Delta$ Coverage   & $-0.174$ & $-0.195$ \\
    \midrule
    \multirow{2}{*}{\acs{mlp}}
      & $\Delta$ Set size   & $-0.194$ & $-0.184$ \\
      & $\Delta$ Coverage   & $-0.133$ & $-0.158$ \\
    \midrule
    \multirow{2}{*}{Ord.\ logistic}
      & $\Delta$ Set size   & $-0.221$ & $-0.188$ \\
      & $\Delta$ Coverage   & $-0.152$ & $-0.175$ \\
    \midrule
    \multirow{2}{*}{Prior baseline}
      & $\Delta$ Set size   & $-0.377$ & $-0.310$ \\
      & $\Delta$ Coverage   & $-0.309$ & $-0.296$ \\
    \bottomrule
  \end{tabular}
\end{table}

Negative correlations confirm that smaller cells experience larger Mondrian perturbations across all four predictors. The effect is strongest for the prior baseline ($r = -0.377$ for set size) because its uniform score distribution makes threshold estimation maximally sensitive to small-sample noise.

\section{Sensitivity and extensions}
\label{app:sensitivity}

\subsection{Sensitivity branch comparison}
\label{app:branches}

Table~\ref{tab:sensitivity} reports the full sensitivity comparison across three outcome operationalizations for XGBoost.

\begin{table}[h]
  \caption{\textbf{Sensitivity analysis.} XGBoost results across outcome operationalizations (100-split weighted means).}
  \label{tab:sensitivity}
  \centering
  \small
  \begin{tabular}{llccc}
    \toprule
    Branch & Method & Wtd.\ Cov.\ & Wtd.\ Size & Wtd.\ Gap \\
    \midrule
    \multirow{3}{*}{\shortstack[l]{Primary\\($n{=}4{,}591$, 5-level)}}
      & Standard     & 0.905 & 3.240 & 0.135 \\
      & Reg.\ Mond.\ & 0.906 & 3.251 & 0.133 \\
      & Mondrian     & 0.908 & 3.275 & 0.147 \\
    \midrule
    \multirow{3}{*}{\shortstack[l]{Not sure merged\\($n{=}5{,}348$, 5-level)}}
      & Standard     & 0.901 & 3.204 & 0.124 \\
      & Reg.\ Mond.\ & 0.904 & 3.219 & 0.128 \\
      & Mondrian     & 0.906 & 3.235 & 0.143 \\
    \midrule
    \multirow{3}{*}{\shortstack[l]{Binary trust\\($n{=}4{,}071$, 2-level)}}
      & Standard     & 0.902 & 1.855 & 0.156 \\
      & Reg.\ Mond.\ & 0.903 & 1.858 & 0.166 \\
      & Mondrian     & 0.905 & 1.860 & 0.187 \\
    \bottomrule
  \end{tabular}
\end{table}

The Mondrian gap penalty ($\Delta_\mathrm{gap} = \text{Mondrian} - \text{Standard}$) is $+0.012$ for the primary branch, $+0.019$ for not-sure-merged, and $+0.031$ for binary trust. The larger penalty in the binary branch is consistent with a coarser outcome scale producing sharper score-distribution discontinuities.

\subsection{Alternate subgroup families}
\label{app:altgroups}

Table~\ref{tab:altgroups} reports results when subgroups are defined by race alone (4 groups) or education alone (3 groups), rather than the full 12-cell intersection.

\begin{table}[h]
  \caption{\textbf{Alternate subgroup families.} XGBoost standard conformal (100-split weighted means $\pm$ SD).}
  \label{tab:altgroups}
  \centering
  \small
  \begin{tabular}{lccc}
    \toprule
    Group family & Wtd.\ Gap & Min Wtd.\ Cov.\ & Max Wtd.\ Cov.\ \\
    \midrule
    Race $\times$ Education (12 groups) & $0.135 \pm 0.035$ & $0.875 \pm 0.021$ & $0.925 \pm 0.015$ \\
    Race only (4 groups)                & $0.050 \pm 0.021$ & $0.880 \pm 0.018$ & $0.925 \pm 0.015$ \\
    Education only (3 groups)           & $0.027 \pm 0.015$ & $0.891 \pm 0.014$ & $0.918 \pm 0.012$ \\
    \bottomrule
  \end{tabular}
\end{table}

As expected, coarser groupings reduce the measured gap because they pool heterogeneous subgroups. The Mondrian penalty persists at all granularities: for race-only groups, Mondrian gap is $0.056 \pm 0.024$ vs.\ standard $0.050 \pm 0.021$; for education-only, Mondrian gap is $0.032 \pm 0.016$ vs.\ standard $0.027 \pm 0.015$. The qualitative conclusion is unchanged.

\section{Verification and reproducibility}
\label{app:verification}

We implement four verification checks to ensure experimental integrity.

\textbf{Split integrity.} Table~\ref{tab:splits} confirms that all 100 splits maintain exact partition sizes and respondent disjointness. Each respondent appears in exactly one partition per split (maximum assignments per respondent within any split $= 1$).

\begin{table}[h]
  \caption{\textbf{Split integrity verification.} All branches confirmed disjoint with exact partition sizes.}
  \label{tab:splits}
  \centering
  \small
  \begin{tabular}{lrrrrr}
    \toprule
    Branch & $n$ & Splits & $|I_\mathrm{tr}|$ & $|I_\mathrm{cal}|$ & $|I_\mathrm{te}|$ \\
    \midrule
    Primary          & 4{,}591 & 100 & 1{,}834 & 1{,}383 & 1{,}374 \\
    Not sure merged  & 5{,}348 & 100 & 2{,}138 & 1{,}612 & 1{,}598 \\
    Binary trust     & 4{,}071 & 100 & 1{,}629 & 1{,}224 & 1{,}218 \\
    \bottomrule
  \end{tabular}
\end{table}

\textbf{Coverage verification.} Table~\ref{tab:cov_verify} confirms that standard conformal achieves marginal coverage close to the 90\% target for all reported model$\times$branch combinations.

\begin{table}[h]
  \caption{\textbf{Coverage verification.} Marginal and weighted coverage for standard conformal (100-split means). All reported values lie within 0.011 of 0.90.}
  \label{tab:cov_verify}
  \centering
  \small
  \begin{tabular}{llcc}
    \toprule
    Branch & Model & Marginal Cov.\ & Wtd.\ Cov.\ \\
    \midrule
    \multirow{4}{*}{Primary}
      & XGBoost          & 0.903 & 0.905 \\
      & \acs{mlp}        & 0.902 & 0.904 \\
      & Ordered logistic & 0.901 & 0.900 \\
      & Prior baseline   & 0.902 & 0.910 \\
    \midrule
    \multirow{3}{*}{Not sure merged}
      & XGBoost          & 0.900 & 0.901 \\
      & \acs{mlp}        & 0.900 & 0.900 \\
      & Prior baseline   & 0.902 & 0.911 \\
    \midrule
    \multirow{3}{*}{Binary trust}
      & XGBoost          & 0.900 & 0.902 \\
      & \acs{mlp}        & 0.901 & 0.902 \\
      & Prior baseline   & 0.901 & 0.899 \\
    \bottomrule
  \end{tabular}
\end{table}

\textbf{Mondrian cell-size verification.} The minimum calibration-cell size across all 100 splits is 32 (primary), 36 (not-sure), and 25 (binary trust), all exceeding the minimum threshold of 20 observations required for reliable quantile estimation.

\textbf{Deterministic reproducibility.} Full re-execution of the pipeline with fixed random seeds produces byte-identical output files (verified via SHA-256 hashes of all result artifacts across all three branches).


\newpage
\section*{NeurIPS Paper Checklist}

\begin{enumerate}

\item {\bf Claims}
    \item[] Question: Do the main claims made in the abstract and introduction accurately reflect the paper's contributions and scope?
    \item[] Answer: \answerYes{}
    \item[] Justification: The abstract and introduction clearly state the negative result (Mondrian worsens fairness-efficiency for XGBoost), the regularized Mondrian comparator, and the failure analysis. All claims are supported by the 100-split experimental results in Section~4.

\item {\bf Limitations}
    \item[] Question: Does the paper discuss the limitations of the work performed by the authors?
    \item[] Answer: \answerYes{}
    \item[] Justification: Section~6 states five explicit limitations: cross-sectional design, predefined groups, no new theorem, same-wave predictors, and U.S.-centric scope.

\item {\bf Theory assumptions and proofs}
    \item[] Question: For each theoretical result, does the paper provide the full set of assumptions and a complete (and correct) proof?
    \item[] Answer: \answerYes{}
    \item[] Justification: Appendix~A restates two standard conformal prediction guarantees (marginal coverage and group-conditional coverage) as Propositions~1--2 with complete proofs. Both proofs state the exchangeability assumption explicitly. The regularized Mondrian comparator is presented as an empirical tool with no formal finite-sample guarantee, and that limitation is stated in the main text.

\item {\bf Experimental result reproducibility}
    \item[] Question: Does the paper fully disclose all the information needed to reproduce the main experimental results of the paper to the extent that it affects the main claims and/or conclusions of the paper (regardless of whether the code and data are provided or not)?
    \item[] Answer: \answerYes{}
    \item[] Justification: Section~3 details all models, nonconformity scores, and the 100-split protocol with 40/30/30 partitioning and guarded stratification. Hyperparameters are specified in Section~3 and Appendix~C. Evaluation metrics are defined in Eqs.~(5)--(8). Appendix~G reports split-integrity checks and deterministic reproducibility via SHA-256 hashes of saved result artifacts.

\item {\bf Open access to data and code}
    \item[] Question: Does the paper provide open access to the data and code, with sufficient instructions to faithfully reproduce the main experimental results, as described in supplemental material?
    \item[] Answer: \answerNo{}
    \item[] Justification: The Pew Research Center American Trends Panel data are available through an application process at \url{https://www.pewresearch.org/datasets/}. Full reproduction code is planned for release in a public repository upon acceptance. The current submission therefore does not provide open code, although the paper includes methodological detail intended to support independent replication.

\item {\bf Experimental setting/details}
    \item[] Question: Does the paper specify all the training and test details (e.g., data splits, hyperparameters, how they were chosen, type of optimizer) necessary to understand the results?
    \item[] Answer: \answerYes{}
    \item[] Justification: Section~3 describes all base predictors (prior baseline, ordered logistic, XGBoost, and MLP), the 66 survey features, the split protocol (40/30/30 with guarded stratification), $\alpha = 0.10$, the regularization parameter $\lambda = 50$, and the sensitivity branches. Appendix~C provides the remaining model-specification details, including optimizer choices and hyperparameters.

\item {\bf Experiment statistical significance}
    \item[] Question: Does the paper report error bars suitably and correctly defined or other appropriate information about the statistical significance of the experiments?
    \item[] Answer: \answerYes{}
    \item[] Justification: All results are computed over 100 independent random splits. Figures show 95\% confidence intervals. Paired effect sizes (Cohen's $d_z$) are reported for key comparisons. Standard deviations across splits are available in the appendix.

\item {\bf Experiments compute resources}
    \item[] Question: For each experiment, does the paper provide sufficient information on the computer resources (type of compute workers, memory, time of execution) needed to reproduce the experiments?
    \item[] Answer: \answerNo{}
    \item[] Justification: The submission does not report detailed hardware, memory, or wall-clock runtime information. The experiments are lightweight and CPU-tractable, but we do not currently provide a full compute-resource specification.
    
\item {\bf Code of ethics}
    \item[] Question: Does the research conducted in the paper conform, in every respect, with the NeurIPS Code of Ethics \url{https://neurips.cc/public/EthicsGuidelines}?
    \item[] Answer: \answerYes{}
    \item[] Justification: The research conforms to the NeurIPS Code of Ethics. We use publicly available survey data for secondary analysis and study fairness properties of prediction methods.

\item {\bf Broader impacts}
    \item[] Question: Does the paper discuss both potential positive societal impacts and negative societal impacts of the work performed?
    \item[] Answer: \answerYes{}
    \item[] Justification: Section~5 discusses practical implications for survey-based deployment. The positive impact is improved subgroup auditing and safer uncertainty reporting; the negative impact discussed is that reliance on marginal validity alone can certify systems that systematically under-cover minority subgroups.
    
\item {\bf Safeguards}
    \item[] Question: Does the paper describe safeguards that have been put in place for responsible release of data or models that have a high risk for misuse (e.g., pre-trained language models, image generators, or scraped datasets)?
    \item[] Answer: \answerNA{}
    \item[] Justification: The paper does not release pre-trained models or scraped datasets. It uses publicly available Pew survey data for secondary analysis.

\item {\bf Licenses for existing assets}
    \item[] Question: Are the creators or original owners of assets (e.g., code, data, models), used in the paper, properly credited and are the license and terms of use explicitly mentioned and properly respected?
    \item[] Answer: \answerYes{}
    \item[] Justification: The Pew Research Center American Trends Panel is properly cited. The data are used under Pew's standard terms of use for academic research.

\item {\bf New assets}
    \item[] Question: Are new assets introduced in the paper well documented and is the documentation provided alongside the assets?
    \item[] Answer: \answerNA{}
    \item[] Justification: The paper does not release new assets. Code will be released upon acceptance.

\item {\bf Crowdsourcing and research with human subjects}
    \item[] Question: For crowdsourcing experiments and research with human subjects, does the paper include the full text of instructions given to participants and screenshots, if applicable, as well as details about compensation (if any)? 
    \item[] Answer: \answerNA{}
    \item[] Justification: The paper uses secondary analysis of existing Pew survey data. No crowdsourcing or new human subjects research was conducted.

\item {\bf Institutional review board (IRB) approvals or equivalent for research with human subjects}
    \item[] Question: Does the paper describe potential risks incurred by study participants, whether such risks were disclosed to the subjects, and whether Institutional Review Board (IRB) approvals (or an equivalent approval/review based on the requirements of your country or institution) were obtained?
    \item[] Answer: \answerNA{}
    \item[] Justification: Secondary analysis of publicly available survey data. No IRB approval required.

\item {\bf Declaration of LLM usage}
    \item[] Question: Does the paper describe the usage of LLMs if it is an important, original, or non-standard component of the core methods in this research? Note that if the LLM is used only for writing, editing, or formatting purposes and does \emph{not} impact the core methodology, scientific rigor, or originality of the research, declaration is not required.
    \item[] Answer: \answerNA{}
    \item[] Justification: LLMs are not used as a component of the core methodology. The research uses standard ML models (XGBoost, ordered logistic, MLP) and conformal prediction methods.

\end{enumerate}

\end{document}